\documentclass[a4paper,USenglish,cleveref,autoref,thm-restate]{lipics-v2021}

\hideLIPIcs
\nolinenumbers

\graphicspath{{./Figures/}}

\usepackage[numbers,sort&compress]{natbib}
\usepackage[T1]{fontenc}
\usepackage{xcolor}

\usepackage{microtype}
\usepackage{mathtools,bm}

\usepackage{hyperref}

\DeclareMathOperator{\len}{len}
\DeclareMathOperator{\vol}{vol}

\def\RR{\ensuremath{\mathbb R}}
\def\NN{\ensuremath{\mathbb N}}
\def\BB{\ensuremath{\mathbb B}}

\def\eps{\ensuremath{\varepsilon}}

\def\origin{\ensuremath{\bm{o}}}
\def\target{\ensuremath{\bm{t}}}

\definecolor{defblueee}{rgb}{0.1,0.4,0.6} 
\def\DEF#1{\textbf{\emph{\textcolor{defblueee}{#1}}}}

\title{Searching in Euclidean Spaces with Predictions\thanks{A preliminary version
of this work was presented at WAOA 2024~\cite{CabelloG24}.}}

\author{Sergio Cabello}{Faculty of Mathematics and Physics, University of Ljubljana, Slovenia\\ Institute of Mathematics, Physics and Mechanics, Slovenia}{sergio.cabello@fmf.uni-lj.si}{https://orcid.org/0000-0002-3183-4126}{}

\author{Panos Giannopoulos}{Department of Computer Science, City, University of London, UK}{Panos.Giannopoulos@city.ac.uk}{https://orcid.org/0000-0002-6261-1961}{}

\authorrunning{S. Cabello and P. Giannopoulos}
\Copyright{Sergio Cabello and Panos Giannopoulos}



\begin{document}

\maketitle

\begin{abstract}
We study the problem of searching for a target at some unknown location 
in $\mathbb{R}^d$ when additional information regarding the position of 
the target is available in the form of predictions.  In our setting, 
predictions come as approximate distances to the target: 
for each point $p\in \mathbb{R}^d$ that the searcher visits, we obtain 
a value $\lambda(p)$ such that $|p\bm{t}|\le \lambda(p) \le c\cdot |p\bm{t}|$, 
where $c\ge 1$ is a fixed constant, $\bm{t}$ is the position of the target, 
and $|p\bm{t}|$ is the Euclidean distance of $p$ to $\bm{t}$. The cost of the search is the length of the path followed by the searcher.
Our main positive result is a strategy that achieves $(10c)^{d+1}$-competitive ratio, 
even when the constant $c$ is unknown. 
We also give a lower bound of roughly $(c/4)^{d-1}$ on the competitive 
ratio of any search strategy in $\RR^d$, assuming that $c\ge 4$.
\end{abstract}


\section{Introduction}

The problem of searching for a target positioned at some unknown location 
in some region is a classic search game problem that has been well-studied 
in both fields of Computational Geometry and Operations Research. The problem 
comes in many different versions, such as linear search (i.e., searching 
for a target on a line)~\cite{Bellman63, Beck64, BeckN70, Gal74}, 
searching in the plane~\cite{Bellman56, Isbell57, Gal80, Baeza-YatesCR_InfComp93}, 
searching in concurrent rays~\cite{Gal72, Gal74, Baeza-YatesCR_InfComp93}, and searching 
inside polygonal regions~\cite{Klein05, Schuierer98}. 
The book by Alpern and Gal~\cite{AlpernG_Book03} provides an extensive overview of general search games, while Ghosh and Klein~\cite{GhoshK_CompSciRev10} survey search problems in planar domains. 

The searcher starts from some given position, follows some path according 
to some strategy until the target, usually a point, is reached or detected, for some appropriate 
definition of ``detection''. In $1$-dimensional settings, e.g., an infinite line, one usually requires 
that the searcher passes through the target point. In the plane one 
usually requires that the searcher is within some distance of the target, or sees the target, if there are obstacles, or that the target lies on the segment connecting the searcher's current and starting positions~\cite{Gal80, GhoshK_CompSciRev10, Langetepe_SODA10}. 

The cost of the search is the length of the path followed by the searcher 
and the objective is to find an efficient search strategy. Efficiency is 
usually measured by the \DEF{competitive ratio}, which, in this setting, 
is the ratio of the length of the path of the searcher to the actual 
Euclidean distance of the starting position to the target. 

In this work, we consider the problem of finding a target point $\target$ in 
Euclidean space $\mathbb{R}^d$ when additional information regarding 
the position of the target is available in the form of \DEF{predictions}.
Here, the predictions are the approximate distance to $\target$ 
for all points visited during the search, e.g., a value between, say, $|p\target|$ and $2|p\target|$
for each point $p$ visited.
Such an estimate could be obtained for example in a scenario where
one takes into account the strength of a signal broadcasted by the target.

Algorithms with predictions is a concept that has been introduced
relatively recently; see the survey by
Mitzenmacher and Vassilvitskii~\cite{MitzenmacherV_BeyondW_CAA20}.
The general idea is that on top of the usual input data 
we are also given additional and possibly inaccurate (noisy) information, the prediction, 
that should assist the algorithm to be more effective.
The improvement in performance depends on the accuracy of the
prediction.

We continue this section with the problem setup, a summary of our contribution together with a short discussion on our predictions model, and related work.

\subsection{Problem setup}
We consider the following search problem in $\RR^d$. 
Assume that there is a fixed but unknown \DEF{target} point~$\target\in \RR^d$.
Without loss of generality, we start the search at the origin, which we denote by~$\origin$. 
We want to find a curve $\gamma$ that starts at~$\origin$ and ends at~$\target$.
The \DEF{cost} of the search is the Euclidean length of the curve $\gamma$.

As we search for the target, we have approximate information about the distance 
to it from each point that we have visited so far.
More precisely, we assume that there is a constant $c\ge 1$  
and an unknown function
\begin{equation}
\label{eq:condition1}
	\lambda\colon\RR^d \rightarrow \RR_{\ge 0} ~\text{ such that }~
    \forall p\in \RR^d~~ |p\target|\le \lambda(p) \le c\cdot |p\target|.
\end{equation}
We refer to such a function $\lambda$ as a \DEF{$c$-prediction} for the target $\target$.
The constant $c$ is the \DEF{prediction factor} of $\lambda$.
See Figure~\ref{fig:lambda} for an example in $d=1$.
Note that for $c=1$, the function $\lambda$ gives the exact distance to the target.

\begin{figure}\centering
	\includegraphics[page=1,scale=.9]{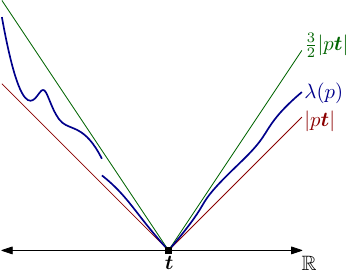}
	\caption{Example of $c$-prediction function $\lambda(p)$ for $c=3/2$. 
		In this example the function is not monotone in $|p\target|$ 
		and it is not continuous.}
	\label{fig:lambda}
\end{figure}

For each point $p$ along the search path we have traversed so far, 
we obtain the value $\lambda(p)$, and the search strategy decides how 
to continue the search depending on that information.
We know when we have reached the target because 
$\lambda(p)=0$ holds only when $p=\target$.

As it is common in search games, we are interested in the competitive ratio of the search
strategy: how does the length of the search path compares to the straight-line
distance from the origin to the target?
To formalize this, for each target $\target\in \RR^d$ and each constant $c\ge 1$, 
we consider the family $\Lambda(\target,c)$ of $c$-predictions for $\target$, 
that is, functions satisfying condition \eqref{eq:condition1}. 
The family of $c$-prediction functions is $\cup_{\target \in \RR^d} \Lambda(\target,c)$.

A search strategy $S$ is \DEF{$\alpha=\alpha(S,c,d)$ competitive} if for all $\target\in \RR^d$ 
and all $\lambda\in \Lambda(\target,c)$, the length of the path
defined by $S$ to reach $\target$ from $\origin$
is at most $\alpha |\origin \target|$. 
Note that we have two possible regimes, depending on
whether $c$ is known or unknown to the search strategy.

\subsection{Our contribution}

Our main contributions are the following:

\begin{itemize}
	\item We introduce a natural, new search problem 
		in $\RR^d$ for $d\ge 1$, under a predictions model where we have approximate information 
		about the distance to the target.
	\item We show that for each dimension $d$ and each
		constant prediction factor $c\ge 1$ there is a search strategy 
		with competitive ratio smaller than $(5c)^{d+1}$. 
		To achieve this, we use $\eps$-nets from metric spaces, also known as $r$-nets, 
		and provide a path of finite length but an infinite number of pieces. 
		This result holds assuming that we know the prediction factor $c$.
		For \emph{unknown} prediction factor, a slightly different search strategy 
		leads to a competitive ratio smaller than $(10c)^{d+1}$. 
		These results are given in Section~\ref{sec:upper_bound}.
	\item We show that for $c\ge 4$,
		any deterministic search strategy in $\RR^d$ with $c$-predictions 
		will have a competitive ratio of at least 
		$(c/4)^{d-1}\cdot \min\{\sqrt{\pi/d},1\}$.
		For this result, we construct an infinite family of $c$-predictions
		and use a volume argument to give a lower bound on the length
		of any path that can discern which $c$-prediction from the family
		is the actual one. The approach is motivated by the techniques
		used to obtain approximation algorithms for 
		the Euclidean Traveling Salesperson with Neighbourhoods.
		However, in our setting we have an infinite number of neighbourhoods,
		one for each $c$-prediction.
		With a slightly worse constant, the lower bound holds also for randomized search strategies.
		This lower bound is shown in Section~\ref{sec:lower_bound}.
\end{itemize}

Additionally, we give several basic properties of the model.
For example, we show that without a prediction function throughout
the whole search, we cannot find the target using a path of bounded length,
and having an infinite number of pieces (segments) in the search path is unavoidable.
We also show that we may assume that the prediction function is continuous. 
We also note that when $c=1$, i.e., we have exact information
about the distance to the target, the target can be reached with
competitive ratio arbitrarily close to $1$. These results are presented in Section~\ref{sec:obs}.

To our knowledge, this is the first search problem in Euclidean spaces with $d\ge 2$
where a search point reaches a target point with constant competitive ratio; 
see the discussion below about related work.
For our upper bounds we use bounds on the cardinality of $\eps$-nets.
For the sake of simplicity in the calculations, we use suboptimal but simple-to-parse
estimates in the simplifications and the cardinality of $\eps$-nets.
Similarly, in our lower bounds we also use volume estimates that are easier to manipulate.
In any case, our upper and lower bounds for the best competitive ratio are still far apart.

Our model, as presented in the continuous setting, is scale-free and general. In the discrete setting, where the target has integral coordinates or where the target can be detected when the searcher is within a given distance from it, our strategies can be easily modified to have a finite number of steps. The lower bound holds in the discrete setting too, albeit with a slightly worse constant. Moreover, by having an approximate distance to the target as the prediction, the model can also accommodate underestimates of the actual distance (in addition to the default overestimates), which can be of interest in practical scenarios. These as well as other extensions are discussed in Section~\ref{sec:Extensions}.


\subsection{Related work}
\label{sec:related_work}
For \emph{linear} search, 
when the exact distance to the target is known, one can easily find the target by walking at most three times this distance. When the distance to the target is unknown, Beck and Newman~\cite{BeckN70} and later Baeza-Yates at al.~\cite{Baeza-YatesCR_InfComp93} showed that a simple doubling strategy has competitive ratio $9$; here, a lower bound on the distance to the target is assumed, otherwise there is no search strategy with constant competitive ratio.  There are also other similar strategies with the same competitive ratio~\cite{GhoshK_CompSciRev10, BoseCD_TCS15}. Moreover, various approaches~\cite{Baeza-YatesCR_InfComp93,BeckN70,HipkeIKL_DAM99,KupavskiiW_DistrComp21} show that $9$ is the best possible competitive ratio for the problem. 
 

Gal~\cite{Gal72, Gal74} introduced the problem of searching for a target on multiple \emph{rays} that 
are concurrent at the starting position and gave an optimal strategy for the case where the distance to the target is unknown. 
This result was rediscovered by Baeza-Yates et al.~\cite{Baeza-YatesCR_InfComp93}.

Baeza-Yates et al.~\cite{Baeza-YatesCR_InfComp93} also considered the problem of finding a target with \emph{integer coordinates} in the plane and presented various search strategies. When the distance to the target is known, it is also easy to get an optimal strategy. However, when the distance is unknown, and with no additional information available, no search strategy can have constant competitive ratio as there are $\Theta(n^2)$ integral points within distance at most $n$ from the origin and any search strategy has to visit all of them in some order. For the natural extension of the problem in $\RR^d$, the latter generalizes to any $d\geq 3$, as there are $\Theta(n^{d})$ integral points within distance at most $n$ from the origin. When we are in $\RR^d$ and the distance to the target is known, we hit a classical problem in Number Theory: on how many ways can we express a positive integer as sum of $d$ squares of integers. For $d\geq 5$, there is no search strategy with constant competitive ratio as there are there are $\Omega(n^{d/2-1})$ integral points at distance exactly $n$ from the origin; see, for example, Vaughan and Wooley~\cite{VaughanW_JRAM18}.

In another variant of the problem in the plane by Gal~\cite{Gal80}, the searcher travels along a path until the target lies on the \emph{segment} connecting the searcher's current and starting positions, essentially sweeping around its starting position with an infinitely elastic cord until the target is swept. For this problem, Gal~\cite{Gal80} gave a spiral search strategy achieving a competitive ratio of $17.289{\ldots}$,  while Langetepe~\cite{Langetepe_SODA10} showed that this ratio is optimal.

Hipke et al.~\cite{HipkeIKL_DAM99} considered linear search when the target is at \emph{distance}
at least $1$ and at most $D\ge 1$ from the starting point, where $D$ is 
known at the start of the search. Bose et al.~\cite{BoseCD_TCS15} 
provided a more careful analysis using the roots of a recursive 
sequence of polynomials and gave better lower and upper bounds on the competitive ratio with dependence on $D$. 
L{\'{o}}pez-Ortiz and Schuierer~\cite{Lopez-OrtizS01} considered also this setting, for the case of concurrent rays. 
Compared to these works, there are two main differences in our work.
Firstly and most importantly, we have a prediction all the way through the search,
while they have a prediction only at the start.
Secondly, we consider the problem in more general settings, namely in $\RR^d$
for arbitrary $d$. An upper bound at the start
does not suffice to find a point when $d\ge 2$. (If $d=2$ and the exact distance
is known, the problem can be easily solved since the target has to lie on a known circle, 
and that is the only additional instance that is solvable.)

Banerjee et al.~\cite{BanerjeeC-AGL_ITCS23} considered 
the problem of finding a target in a \emph{graph} with information about the distance
to the target. In their model, the target is at one vertex of the graph 
and at each vertex we have a value stored that is made available only 
when we are adjacent to the vertex.
For most vertices, the value stored at a vertex is the true distance 
from the vertex to the target, but for some vertices the value is wrong. 
Contrary to our setting, they do not assume a bounded error for the information 
at each vertex, but that the information is wrong at 
a bounded number of nodes. The bound on the number of nodes 
with wrong information then appears in the bound for the length of the 
search path. 

Finally, Angelopoulos~\cite{Angelop_ITCS21, Angelop_MFCS23} gave strategies for linear and multiple-ray search under a different model where a one-off, possibly erroneous hint or prediction on the target's position is given at the start of the search. The prediction can be positional, directional, or, in general, a $k$-bit string encoding answers to $k$ binary queries and the measure of the performance of a strategy is a trade-off between the competitive ratio under error-free prediction and that under erroneous prediction.

\section{Notation and preliminaries}
Since the dimension $d$ is always fixed and clear from the context, 
we drop in the notation the dependency on $d$.

The \DEF{ball} centered at $p\in \RR^d$ with radius $r$ is 
$B(p,r)=\{ q\in \RR^d\mid |pq|\le r \}$.
We will also consider the \DEF{spherical shells}
$S(p,r_1,r_2)= \{ q\in \RR^d\mid r_1\le |pq|\le r_2 \}$.
A spherical shell in the plane is an annulus. For a $c$-prediction function $\lambda$, 
whenever we are at a point $p\in \RR^d$, we get a prediction $\lambda(p)$
and we deduce that the target point $\target$ lies in the spherical shell
$S(p,\lambda(p)/c, \lambda(p))$. 
See Figure~\ref{fig:annulus_and_net}.


Note that we have made a modeling decision, namely
we have assumed that there is an unknown function $\lambda$ such that
at each point $p$ we get the prediction $\lambda(p)$.
More generally, it could happen that we visit the same point $p$ multiple
times and at each time we get a different estimate of the distance from $p$
to the target. However, getting a different estimate can only help,
as it provides more information: if we get two
different $c$-predictions $\lambda$ and $\lambda'$ at different times 
at the same point $p$, then we know that
the target lies in the spherical shell 
$S(p,\max\{\lambda,\lambda'\}/c, \min\{\lambda,\lambda'\})$.
This is more information than what we get if $\lambda'=\lambda$
because then we can only conclude that the target lies 
in $S(p,\lambda/c,\lambda)$, which is strictly larger.
Thus, when searching for an optimal search strategy, 
we can assume that each time we visit the same point we get the same 
prediction. In particular, a search strategy could simply ignore 
the information obtained in the second and subsequent visits 
to the same point.

A \DEF{path} in $\RR^d$ is a continuous function 
$\pi\colon [0,1] \rightarrow \RR^d$. The paths in our search strategies
will consist of an infinite number of straight-line segments
and will exhibit a Zeno-like phenomena: 
they make an infinite number of turns in finite time and length.
To show that the paths we define reach the target, 
we will use the following property,
whose proof is an standard argument in continuity.

\begin{lemma}
\label{le:continuous}
	Let $\pi\colon [0,1] \rightarrow \RR^d$ be a path.
	Assume that there is a point $p\in \RR^d$ with the following property:
	for each $\eps>0$ there exists some $\delta\in (0,1]$
	such that the subpath $\pi([1-\delta,1])$ is contained in $B(p,\eps)$.
	Then $\pi(1)=p$, that is, $p$ is the endpoint of the path $\pi$.
\end{lemma}

We will use $\eps$-nets from metric space theory: 
An \DEF{$\eps$-net} for the ball $B(p,r)$ is a subset $N$ 
of points from $B(p,r)$ such that 
(i) each point of $B(p,r)$ is at distance at most $\eps$ 
from some point of $N$, and
(ii) each two distinct points of $N$ are at distance at least $\eps$.
Condition (i) can be equivalently be stated as
$B(p,r)\subseteq\bigcup_{q\in N}B(q,\eps)$.
Condition (ii) is equivalent to telling
that the balls $B(q,\eps/2)$, where $q\in N$,
are pairwise interior disjoint.

The following bound on $\eps$-nets follows from a well-known 
technique using volumes. We will not use the lower bound,
but we include it to show that the upper bound is a reasonable
estimate.

\begin{lemma}
\label{lem:dense}
	In $\RR^d$, for each $\eps\le r/2$,
	a ball of radius $r$ has a $\eps$-net
	with at least $(r/\eps)^d$ and at most $(5r/2\eps)^d$ elements.
\end{lemma}

\begin{proof}
	Using scaling and translation, it suffices to show
	the result for the unit ball centered at the origin, 
	$B(\origin,1)$, and for $\eps\le 1/2$.
	
	Let $N$ be an inclusion-wise maximal subset of points 
	from $B(\origin,1)$ such
	that any two points are at distance at least $\eps$.
	We claim that $N$ is a $\eps$-net.
	By definition, $N$ satisfies property (ii).
	Since by maximality we cannot add any other point of 
	$B(\origin,1)$ to $N$, each point of $B(\origin,1)$ 
	has some point of $N$ at distance at most $\eps$, 
	and thus property (i) is also satisfied.

	It remains to bound from below and above the cardinality of $N$.
	Because of property (i), the ball $B(\origin,1)$ is contained in
	$\bigcup_{q\in N} B(q,\eps)$.
	Using the constant $V_d=\tfrac{\pi^{d/2}}{\Gamma(d/2+1)}$ 
	that gives the volume of the unit ball in $\RR^d$,
	we then have 
	\[
		|N| \ge \frac{\vol(B(\origin, 1))}{\vol(B(\origin,\eps))} =
			\frac{V_d}{V_d\cdot \eps^d}
			= \frac{1}{\eps^d}.
	\]
	On the other hand, property (ii) implies that the family 
	of balls $\{ B(q,\eps/2) \mid q\in N\}$ are 
	pairwise interior disjoint. Each such ball is contained in the 
	enlarged ball $B(\origin, 1+\eps/2)$, and therefore
	\[
		|N| \le \frac{\vol(B(\origin, 1+\eps/2))}{\vol(B(\origin,\eps/2))} =
			\frac{V_d \cdot \big( (1+\eps/2)^d\big)}{V_d\cdot (\eps/2)^d}
			= \left(1+\frac{2}{\eps}\right)^d \le 
			\left( \frac{5}{2\eps} \right)^d,
	\]
	where in the last inequality we used that $\eps\le 1/2$.
\end{proof}


\section{Observations about the problem setup}
\label{sec:obs}

First, it is important to note that a $2$-dimensional disk cannot be covered 
by a curve of bounded length. 
This is definitely not surprising, but it requires a proof,
as the so-called space filling curves~\cite{Sagan_Book} 
can cover a $d$-dimensional body for $d\ge 2$.
Swanepoel~\cite{Swanepoel} provided 
a simple proof
for the following result.

\begin{lemma}
\label{lem:bound}
	For each $r>0$ and each $p\in \RR^2$, 
	any curve that passes through each point of the $2$-dimensional
	disk $B(p,r)$ has infinite length.
\end{lemma}

\begin{proof}
	Consider a rectangular, regular grid of $n\times n$ points inside $B(p,r)$ 
	with total side length $\Theta(r)$; the distance between any
	two neighbour grid points is $\Omega(r/n)$.
	Any curve that goes through those $n^2$ points has length 
	$\Theta(rn)$ because the Euclidean minimum spanning tree
	for those points has length $(n^2-1)\cdot \Theta(r/n)= \Theta(rn)$.
	Sine $n$ can be chosen arbitrarily large, for any number $L$, 
	we can find a set of points that cannot be in the image 
	of any path with length $L$.
\end{proof}

This result implies that, without additional 
information or constraints, a searcher cannot reach a target point
$\RR^d$ for $d\ge 2$, if it only knows that the target lies
inside a ball of positive radius, no matter how small it is.

We say that a $c$-prediction function $\lambda\in \Lambda(\target,c)$ 
is \DEF{open} if for all $p\neq \target$
we have $|p\target| < \lambda(p) < c\cdot |p\target|$. In other words,
at each non-target point the prediction is not the correct distance and it is not 
the maximum possible it could be. This implies that $\target$ lies
in the \emph{interior} of the spherical shell $S(p,\lambda(p)/c, \lambda(p))$,
which is an open set. For example, for $c>1$
the function $\lambda(p)=\tfrac{1+c}{2} \cdot |p\target|$
is open because $1< \tfrac{1+c}{2} < c$.
Note that the definition is meaningful only for $c>1$.
In the following we show that, in general, until we do not reach the
target point, there is a ball such that any point in the ball can still be the target.

\begin{proposition}
\label{pro:flexibility}
	Assume that we have explored a compact set $X$ of points of $\RR^d$, 
	that the target $\target$ does not belong to $X$ and that 
	the $c$-prediction $\lambda$ is open and continuous.
	There exist an $\eps>0$ such that the ball 
	$B(\target,\eps)$ has the following property: 
	for each $\target'\in B(\target,\eps)$, there is a $c$-prediction 
	$\lambda'$ in $\Lambda(\target',c)$ that is continuous, open and
	such that $\lambda'(p)=\lambda(p)$ for all $p\in X$.
\end{proposition}

Let us explain the necessity to 
assume that the prediction is \emph{open}. This is not because the problem becomes easy otherwise. In general, without this restriction, predictions can still be such that there is an ``uncertainty ball'' which the target can lie anywhere in. However, it could also happen that for a \emph{finite} set $X$ of points
we have that $\cap_{p\in X}S(p,\lambda(p)/c,\lambda(p))$ is a single point,
which must then be the target $\target$. (This cannot happen when all predictions are open.) In such a case, if the searcher passes 
through all points of $X$, they have not reached the target but already have 
deduced its position. 
We further assume that $X$ is compact and $\lambda$ is continuous 
to avoid that such a deduction can be made when $X$ is infinite. 

\begin{proof}[Proof of Proposition~\ref{pro:flexibility}]
	We first provide a simple extension tool.
	Let $S$ be a sphere and let $\target'$ be a point in the interior of the ball bounded by $S$.
	Assume that we have a continuous and open $c$-prediction $\lambda\colon S\rightarrow \RR$
	for the target $\target'$, but defined only on $S$.
	Then there is a continuous and open $c$-prediction $\lambda'$ 
	for the target $\target'$ 
	defined on the whole $\RR^d$ that extends $\lambda'$, that is, 
	$\lambda(p)=\lambda'(p)$ for all $p\in S$.
    To show this, we introduce the following notation:
	for each point $p\neq\target'$, let $\rho(\target',p,S)$ be the unique point 
	where the ray with origin at $\target'$ through $p$ intersects the sphere $S$.
	See Figure~\ref{fig:flexibility}, left.
    We then define 
	$\lambda'(p):=\frac{|p\target'|}{|\rho(\target',p,S)\target'|} \cdot \lambda(\rho(t',p,S))$
	for all $p\in \RR^d\setminus\{ \target'\}$, and $\lambda'(\target'):=0$.
	It is easy to see that $\lambda'$ is an open $c$-prediction for $\target'$
	because at each ray	from $\target'$ it is just a linear function.
	It is also easy to see that $\lambda'$ is continuous because $\lambda$ is continuous.
	Finally, it is clear from the definition that $\lambda'$ extends $\lambda$
	because for each $p\in S$ we have $p=\rho(\target',p,S)$.

	\begin{figure}\centering
		\includegraphics[page=6,scale=.75]{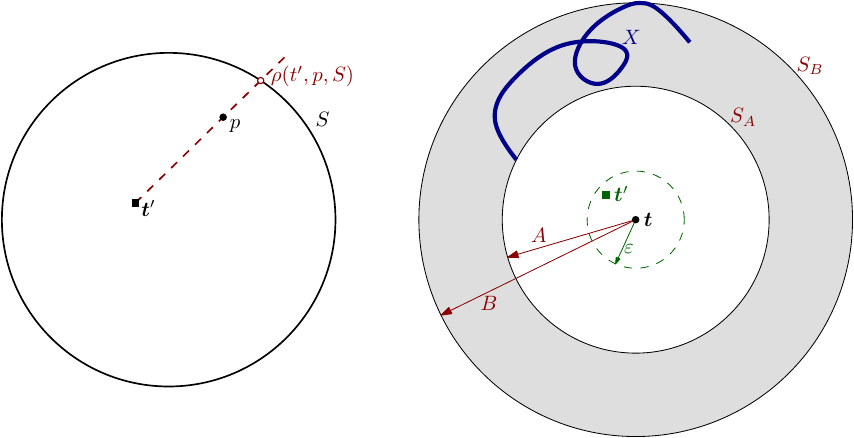}
		\caption{Proof of Proposition~\ref{pro:flexibility}. Left: definition of $\rho(\target',p,S)$
			to extend $\lambda$ defined on $S$. Right: Main part of the proof.}
		\label{fig:flexibility}
	\end{figure}

	We now turn our attention to the main statement.
	Fix a $c$-prediction $\lambda$ for the target $\target$	that is open and continuous.
	Figure~\ref{fig:flexibility}, right, may help to follow the notation.
	Since $X$ is compact and the distance function is continuous, 
	we can define the values
	\[
		A:= \min \{ |p\target|\mid p\in X \} 
		~~\text{ and }~~
		B:= \max \{ |p\target|\mid p\in X \}.
	\]
	They are both positive because $\target\notin X$, 
	and we have $X\subseteq S(\target,A,B)$.
	Because $\target\notin S(\target,A,B)$ and $S(\target,A,B)$ is compact,
	we can also define the values 
	\[
		\eps':=\min \{ \lambda(p)-|p\target| \mid p\in S(\target,A,B) \} 
		~~\text{ and }~~
		\eps'':= \min \{ |p\target| - \lambda(p)/c \mid p\in S(\target,A,B) \}.
	\]
	Because $\lambda$ is open, both values $\eps'$ and $\eps''$ are strictly positive. 
	Let $\eps=\min \{ \eps',\eps''\}/2$, which is also strictly positive.
	We thus have
	\[
		\forall p\in S(\target,A,B):~~~ 
		|p\target|+\eps < \lambda(p) < c \cdot (|p\target| - \eps).
	\]
	We will show that all the points in the ball $B(\target,\eps)$ 
	are possible targets. This is natural because each of those points
	is consistent with the information we got:
	the disk $B(\target,\eps)$ is contained in the spherical shell
	$S(p,\lambda(p)/c,\lambda(p))$ for all $p\in X$ (and actually 
	all $p\in S(\target,A,B)$).
	
	Note that for each $p\in S(\target,A,B)$ and each $\target'\in B(\target,\eps)$
	we have
	\[ 
		\lambda(p) ~<~ c \cdot (|p\target| - \eps) ~\le~ 
			c \cdot (|p\target'|+|\target' \target| - \eps) ~\le~
			c \cdot (|p\target'|)
	\]
	and
	\[ 
		\lambda(p) ~>~ |p\target| + \eps ~\ge ~ |p\target'|-|\target'\target| +\eps
			~\ge~ |p\target'|.
	\]
	Therefore the restriction of $\lambda$ to the spherical shell $S(\target,A,B)$
	is also an open and continuous $c$-prediction for $\target'$.
	
	Let $S_A$ and $S_B$ be the inner and outer spheres bounding
	the spherical shell $S(\target,A,B)$, respectively.
	Let $\lambda_A$ and $\lambda_B$ be the restriction of $\lambda$ 
	to $S_A$ and $S_B$, respectively. They are open and continuous because
	they are a restriction of $\lambda$.

	Fix any point $\target'\in B(\target,\eps)$.
	As discussed at the beginning of the proof,
	we can extend $\lambda_A$ to an open and continuous $c$-prediction 
	$\lambda'_A$ for the target $\target'$. 
	Similarly, we can extend $\lambda_B$ to an open and continuous 
	$c$-prediction $\lambda'_B$ for the target $\target'$.
	
	Finally, we can define the function $\lambda''\colon \RR^d\rightarrow \RR$
	by 
	\[
		\lambda''(p) ~=~ \begin{cases}
			\lambda'_A(p), &\text{if $p\in B(\target,A)$, ~~~[inside]}\\
			\lambda(p), &\text{if $p\in S(\target,A,B)$, ~~~[spherical shell]}\\
			\lambda'_B(p), &\text{if $p\notin B(\target,B)$ ~~~[outside].}\\
			\end{cases}
	\]
	This function $\lambda''$ is continuous because $\lambda'_A(p)=\lambda(p)$ 
	for all $p\in S_A$, $\lambda'_B(p)=\lambda(p)$ for all $p\in S_B$,
	and because $\lambda,\lambda'_A,\lambda'_B$ are continuous.
	It is also open and a $c$-prediction because each of the functions
	$\lambda,\lambda'_A,\lambda'_B$ has these properties.
	Finally, it is obvious that $\lambda(p)=\lambda''(p)$ for all $p\in X$
	because $X\subseteq S(\target,A,B)$.	
\end{proof}

Proposition~\ref{pro:flexibility} implies that
any search strategy that reaches the target for all $c$-predictions, 
where $c>1$, cannot consist of a finite number of pieces.
Indeed, for any $c$-prediction $\lambda$ that is open and continuous,
until we do not reach the target, there is a ball $B_\eps$ of radius $\eps>0$ 
such that the target may be any point of $B_\eps$, and the information
collected so far cannot distinguish among the possible targets. 
In particular, an adversary could change from one continuous and open
$\lambda\in \Lambda(\target,c)$ to another continuous and open
$\lambda'\in \Lambda(\target',c)$, for a suitable $\target'\neq \target$,
at any time before reaching the target because $\lambda$ and $\lambda'$ 
agree on the points that have been explored so far and are indistinguishable.
In particular, any search path needs to have an infinite number of pieces
because an adversary can change the target at any time during the search. 
Finally, knowing that the prediction $\lambda$ is continuous does not help.

Note also that, because of Lemma~\ref{lem:bound},
we cannot cover all those candidate targets of $B_\eps$ 
with a curve of bounded length, unless we collect additional information.
This means that we cannot have a search strategy that ignores
the prediction $\lambda$ after some time, because from that moment
all points in a ball of positive radius keep being possible targets,
and we cannot visit all of them with a curve of bounded length
that neglects additional information.

The following property shows that from the data
collected at any given point, we can infer another $c$-prediction
function using the triangular inequality.
Moreover, the part that we infer is $1$-Lipschitz and thus continuous.
The following result is not used anymore in our work, but it seems
a useful remark for future research.

\begin{lemma}
\label{lem:Lipschitz}
	Let $X\subset \RR^d$ be any non-empty set of points and $\lambda$ a $c$-prediction 
	for target $\target$.
	Then, the function $\tilde\lambda$ defined by 
	\[
		p\in \RR^d ~\mapsto~ \tilde\lambda(p)= \inf\{ |pp'|+\lambda(p') \mid  p'\in X \}
	\]
	is $1$-Lipschitz and the function
	\[
		p\in \RR^d ~\mapsto~ \min \{ \lambda(p), \tilde\lambda(p)\}
	\]
	is a $c$-prediction for $\target$.
\end{lemma}

\begin{proof}
	The idea is that for any $p'\in X$, 
    we can replace the prediction at $p$ by $\min\{|pp'|+\lambda(p'), \lambda(p) \}$,
    which is something we can deduce because of the triangular inequality.

	We first show that the function $\tilde\lambda$ is $1$-Lipschitz.
	Note that the minimum is not necessarily attained because we do not
	assume anything about the prediction function $\lambda$ or the set $X$.
	For each $\eps>0$ and each $p\in \RR^d$, 
    there exists a point $p'_\eps\in X$ such that
	\[
		|pp'_\eps|+\lambda(p'_\eps) ~\le~ \tilde\lambda(p) + \eps.
	\]
	Then
	\begin{align*}
		\forall p,q\in \RR^d:~~~~
		  \tilde\lambda(q) - \tilde\lambda(p) & \le
		  \big( |qp'_\eps| + \lambda(p'_\eps)\big) - \big( |pp'_\eps| + \lambda(p'_\eps) -\eps \big)\\
		  &= |qp'_\eps| - |pp'_\eps| +\eps \\
		  &\le |qp|+\eps.
	\end{align*}
	Because of symmetry we obtain that
	\[
		\forall \eps>0, \forall p,q\in \RR^d:~~~~
		  |\tilde\lambda(p) - \tilde\lambda(q)| \le |pq|+\eps.
	\]
	This implies that 
	\begin{align*}
		\forall p,q\in \RR^d&:~~~~
		  |\tilde\lambda(p) - \tilde\lambda(q)| \le |pq|
	\end{align*}	
	and therefore $\tilde\lambda$ is $1$-Lipschitz.
	Note that, in general, $\tilde\lambda$ is \emph{not} a $c$-prediction
	function because, if the target does not belong to $X$,
	it is always non-zero.

	Next we show that the function 
	$\mu(p):=\min \{ \lambda(p), \tilde\lambda(p)\}$
	is a $c$-prediction for $\target$.
	For this we use that $\lambda$ is a $c$-prediction for $\target$.
	On the one hand we have
	\[
		\forall p\in \RR^d:~~~ \mu(p) ~=~ 
		\min\{ \lambda(p), \tilde\lambda(p) \} 
		~\le~ \lambda(p) ~\le~ c\cdot |p\target|.
	\]
	On the other hand
	\begin{align*}
		\forall p\in \RR^d:~~~ \tilde\lambda(p) ~&=~ 
			\inf\{ |pp'|+\lambda(p') \mid  p'\in X \} 
			~\ge~ \inf\{ |pp'|+|p'\target| \mid  p'\in X \} \\
			~&\ge~ \inf\{ |p\target| \mid  p'\in X \} 
			~\ge~ |p\target|,
	\end{align*}
	and therefore
	\begin{align*}
		\forall p\in \RR^d:~~~ \mu(p) ~=~ 
			\min \{ \lambda(p), \tilde\lambda(p) \} 
			~\ge~ \min\{ |p\target|, |p\target| \} ~=~ |p\target|.
	\end{align*}
	This shows that $\mu(\cdot)$ is a $c$-prediction for $\target$, as claimed
\end{proof}

The relevance of the result is that at any given point $p$ we get two predictions,
$\lambda(p)$ and the value $\tilde\lambda(p)$ that we can infer. 
When they are distinct, we can infer more information, which can help the search.
However, as we mentioned, we do not exploit this in the strategies
we describe.

Finally, we note that the case of $c=1$ is easy.

\begin{proposition}[Case $c=1$]
	For $1$-predictions in $\RR^d$,
	there is a search strategy with competitive ratio $1+\eps$ for each $\eps>0$.
\end{proposition}
\begin{proof}
	As usual, we assume that the search starts at the origin $\origin$.
	When $c=1$, whenever we query for the distance at a point $p$,
	we get $\lambda(p)=|p\target|$ and we conclude that
	the target lies on the $d-1$ sphere  
	centered at $p$ with radius $\lambda(p)$.
	In particular, $\lambda(\origin)=|\origin\target|$.
	
	We have to find a strategy that reaches the target using
	a path of length at most $(1+\eps)|\origin\target|=(1+\eps)\lambda(\origin)$.
	If $\lambda(\origin)=0$, we are already at the target.
	We thus assume that $\lambda(\origin)\neq 0$.

	For each axis,
	we move $\delta=\lambda(\origin)\cdot\eps/(2d)$ along the axis, 
	query the distance
	to the target, and go back to the origin. The target has to lie
	in the intersection of the $d$ spheres centered at the points
	where we queried and also at the sphere centered at the origin.
	The intersection of these $d+1$ spheres is always a single point,
	which must be $\target$.
	The uniqueness of the point can be seen for example using the lifting to 
	the paraboloid in $\RR^{d+1}$ defined by the function 
	$(x_1,\dots,x_d)\in \RR^d \mapsto 
	\left( x_1,\dots,x_d, \sum_{i=1}^d (x_i)^2\right) \in \RR^{d+1}$.
	Each of those spheres corresponds to the intersection
	of a hyperplane with the paraboloid, and the normals of those $d+1$ hyperplanes
	are linearly independent. Therefore, the intersection of those
	hyperplanes is a single point, which by construction must lie on 
	the paraboloid.

	We have walked exactly $(2d)\delta=\eps |\origin \target|$, came back to the origin,
	and we know now the exact position of the target. Walking to the target takes
	additional $|\origin \target|$ length.
\end{proof}

\section{Upper bound}
\label{sec:upper_bound}

In this section we provide search strategies to reach the target
in $\RR^d$ when we have a $c$-prediction. 
We first provide the key lemma that tells us how to get
a sequence of points whose $\lambda$-values decreases geometrically.
We then provide a strategy when the prediction factor is known,
and then discuss how to handle the case for unknown prediction
factor. In this setting, we adapt the notation so that $c^*$
is the true prediction factor, while $c$ is a guess for the
true prediction factor.
At the end of the section we discuss some extensions.

\begin{figure}\centering
	\includegraphics[page=3,scale=.85]{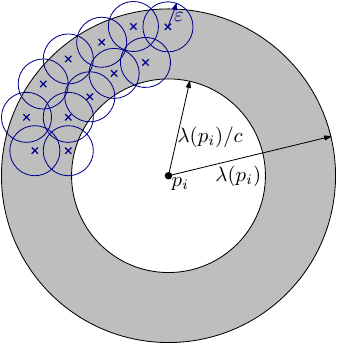}
	\caption{Spherical shell $S(p_i,\lambda(p_i)/c,\lambda(p_i))$ where $\target$ must lie and (part of) a net for the shell.}
	\label{fig:annulus_and_net}
\end{figure}

\begin{lemma}
\label{lem:onestep}
	Assume that we are at point $p_i$ and the prediction factor $c^*$ is perhaps unknown.
	Let $c\ge 1$ be a guess for $c^*$.
	Using a search through a path $\gamma_{i+1}$ of length at most
	$2(5c)^d\cdot \lambda(p_i)$ we get to one of the following outcomes:
	\begin{itemize}
		\item we move from $p_i$ to a point $p_{i+1}$ such that
			$\lambda(p_{i+1})\le \lambda(p_i)/2$, or
		\item we come back to $p_i$ and correctly deduce that $c<c^*$. 
	\end{itemize}
	Moreover, all points of the path $\gamma_{i+1}$ are at distance at most 
	$2\lambda(p_i)$ from $\target$.		
\end{lemma}
\begin{proof}
	Recall that $|p_i\target| \le \lambda(p_i)$.
	Set $\eps=\tfrac{\lambda(p_i)}{2c} \le \tfrac{\lambda(p_i)}{2}$ and
	let $N$ be a $\eps$-net for the ball 
	$B(p_i,\lambda(p_i))$ given by Lemma~\ref{lem:dense}.
	See Figure~\ref{fig:annulus_and_net}.
	Note that 
	\[
		|N|	\le  \left( \frac{5\lambda(p_i)}{2\eps}\right)^d 
			\le (5c)^d.
	\]
	
	Since the target is contained in $B(p_i,\lambda(p_i))$ and 
	$N$ is a $\eps$-net, there is some point $p_*\in N$ such that 
	$|p_*\target|\le \eps = \lambda(p_i)/2c$. 
	If $c^*\le c$, then we have $\lambda(p_*) \le c^*\cdot |p_*\target| \le 
	c^* \cdot \lambda(p_i)/2c \le \lambda(p_i)/2$. 
	If $c^* > c$, then we have no guaranteed useful bound for $\lambda(p_*)$.
	
	We take a path that goes through the points of $N$ in 
	arbitrary order and, at each point of $N$, we query for the prediction $\lambda(\cdot)$.
	We finish the path as soon as we reach some point $q_*\in N$
	such that $\lambda(q_*)\le \lambda(p_i)/2$. If for all the points $p$ of $N$
	we have $\lambda(p)> \lambda(p_i)/2$, then we go back to the point $p_i$.
	This finishes the description of the path $\gamma_{i+1}$.

	In the first case, we set $p_{i+1}=q_*$ and we have moved to
	a point $p_{i+1}$ with $\lambda(p_{i+1})\le \lambda(p_i)/2$.
	When $c^*\le c$, we have to be in this case since the
	point $p_*$ satisfies the stopping condition 
	$\lambda(p_*)\le \lambda(p_i)/2$.
	Thus, if we do not have such a point, we can conclude that 
	$c^*>c$.
	
	To bound the length of $\gamma_{i+1}$, we note that any two points in 
	$B(p_i,\lambda(p_i))$ are at distance at most $2\lambda(p_i)$. We further
	note that the first edge has length at most $\lambda(p_i)$ and, if
	$\gamma_{i+1}$ comes back to $p_i$, also the last edge has length at
	most $\lambda(p_i)$.
	Using the bound of Lemma~\ref{lem:dense} for $|N|$, 
	we get that the path $\gamma_{i+1}$ has length at most 
	\[
		\lambda(p_i) + (|N|-1)\cdot(2\lambda(p_i)) + \lambda(p_i) 
			\le  2\cdot |N|\cdot \lambda(p_i)
		 = 2(5c)^d \cdot \lambda(p_i).
	\]
	Finally, we note that the whole path $\gamma_{i+1}$ is contained in $B(p_i,\lambda(p_i))$,
	which is contained in $B(\target,2\lambda(p_i))$ because 
	$\target\in B(p_i,\lambda(p_i))$. 
	
	Note: It is known~\cite{few1955shortest} that for any set of $n$ points 
	in the $d$-dimensional ball $B(p,r)$ there is a tour of length 
	$r\cdot O\big(n^{\frac{d-1}{d}}\big)$ visiting them. 
	This implies that we can also use a path $\gamma_{i+1}$ of length 
	$\lambda(p_i) \cdot O\big( |N|^{\frac{d-1}{d}}\big) = 
	O((5c)^{d-1})\cdot\lambda(p_i)$. With this, the dependency on $c$ is slightly better
	at the expense of having more ugly-looking constants hidden
	in the $O$-notation.
\end{proof}

\begin{theorem}[Known $c$]
\label{thm:knownc}
	Consider the search with predictions problem in $\RR^d$ where
	the prediction factor $c^*>1$ is known.
	There is a search strategy to reach the target 
	with competitive 
	ratio $4\cdot 5^d\cdot (c^*)^{d+1}$.
\end{theorem}
\begin{proof}
	Let $p_0=\origin$ be the starting point and recall that
	$\lambda(p_0) \le c^* \cdot |\origin \target|$.
	
	For $i=0,1,2\dots$ iteratively, we use Lemma~\ref{lem:onestep} 
	with the guessed prediction factor $c=c^*$ to obtain a path $\gamma_{i+1}$.
	As the prediction factor is correct, we always have the 
	outcome in the first item: 
	$\gamma_{i+1}$ finishes at a point $p_{i+1}$ 
	with $\lambda(p_{i+1})\le \lambda(p_i)/2$.
	It follows by induction that for each $i\in \NN\cup \{ 0 \}$ we have
	$\lambda(p_i)\le \lambda(p_0)/2^i$ and therefore 
	\[
		\len(\gamma_{i+1}) ~\le~ 2(5c^*)^d\cdot \lambda(p_i) 
		~\le~ 2(5c^*)^d\cdot \frac{\lambda(p_0)}{2^i}.
	\]
	Let $\gamma$ be the concatenation of the paths $\gamma_1,\gamma_2,\dots$.
	Then 
	\begin{align*}
		\len(\gamma) ~&=~ \sum_{i=0}^\infty \len(\gamma_{i+1}) 
			~\le~ \sum_{i=0}^\infty \left( 2(5c^*)^d\cdot \frac{\lambda(p_0)}{2^i} \right)
			~=~ 4(5c^*)^d \cdot \lambda(p_0) \\ 
			&\le~ 4 \cdot 5^d (c^*)^{d+1} \cdot |\origin\target|.
	\end{align*}
	The path makes an infinite number of straight-line moves.
	Since for each $i\in \NN$, the suffix of the path $\gamma$ after $p_{i+1}$
	is at distance at most $2\cdot \lambda(p_{i+1}) \le 
	2 \cdot \lambda(p_0)/2^{i+1}= \lambda(p_0)/2^i$ 
	from $\target$, Lemma~\ref{le:continuous} implies that $\target$ 
	is the endpoint of the path $\gamma$.
\end{proof}

\begin{theorem}[Unknown $c$]
\label{thm:unknownc}
	Consider the search with predictions problem in $\RR^d$ where
	the prediction factor $c^*$ is unknown.
	There is a search strategy to reach the target with competitive 
	ratio $6\cdot 10^d \cdot (c^*)^{d+1}$.
\end{theorem}
\begin{proof}
	The basic idea is using an exponential search for the constant $c$.
	At each step we use Lemma~\ref{lem:onestep} to either move to 
	a point with smaller predicted distance to the target or to detect 
	that our guess for $c$ is too small and double it.
	Index $j$ parameterizes the current guess $c=2^j$, 
	and $p^{(j)}_i$ denotes a point during that guess.
	At each step we will increase either $i$ or $j$.

	We start setting $j=1$, $i=0$ and $p^{(j)}_i=p^{(1)}_0=\origin$. 

	From the current point $p^{(j)}_i$, 
	we use Lemma~\ref{lem:onestep} with the guessed prediction factor 
	$c=2^j$ to obtain a path $\gamma^{(j)}_{i+1}$.
	If the outcome is given by the first item of Lemma~\ref{lem:onestep},
	then we get to a point $p^{(j)}_{i+1}$ with 
	$\lambda(p^{(j)}_{i+1})\le \lambda(p^{(j)}_i)/2$; in
	this case we increase $i$.
	If the outcome is given by the second item of Lemma~\ref{lem:onestep},
	then we get back to $p^{(j)}_i$; in this case
	we set $p^{(j+1)}_i=p^{(j)}_i$, increase $j$ and,
	from this point on, we will use the new guessed prediction factor $2^j$,
	which is twice larger than before.
	See Figure~\ref{fig:unknownc} for an schematic view of the sequence
	of points.
	
	\begin{figure}[htb]
	\centering
		\includegraphics[page=4,scale=1]{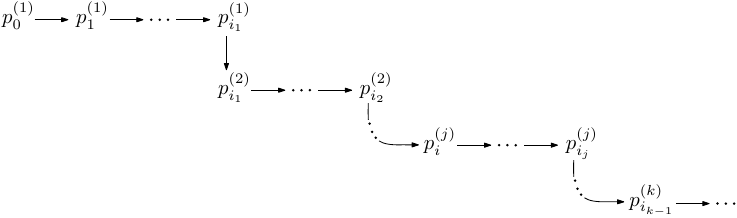}
		\caption{Visualizing the sequence of points 
			$\Big( \bigl( p^{(j)}_i \bigr)_{i\in I_j}\Bigr)_{j\in J}$.
			When increasing $i$ we move right, when increasing $j$ we move down.
			Here, we are using $k=\max(J)$
			and $i_j=\max(I_j)$ for all $j\in J$.}
		\label{fig:unknownc}
	\end{figure}
	
	Let $k$ be the largest value that index $j$ takes through the procedure.
	Note that $k\le \lceil \log_2 c^* \rceil$
	since $2^{\lceil\log_2 c^* \rceil}$ is an overestimate to $c^*$.
	If we arrive to $k= \lceil \log_2 c^* \rceil$, from that
	point on we will always extend the search path using the outcome
	in the first item of Lemma~\ref{lem:onestep}.
	Set $J=\{ 1,\dots, k\}$, which is the set of values that index $j$ takes
	through the procedure.

	For each $j\in J$, let $I_j$ be the set of indices $i$
	such that the point $p^{(j)}_i$ is defined and let $i_j=\max(I_j)$.
	Note that $i_k$ is undefined because $I_k$ is infinite,
	but $i_j$ is defined for all $j<k$ because $I_j$ is finite for all $j<k$.
	By construction, the index $i_j$ is the first element of $I_{j+1}$,
	for all $j<k$. 
	It may happen that, for some $j<k$, the set $I_j$ contains a single element.
	This happens when $j$ is increased in successive steps and thus 
	$p^{(j)}_i=p^{(j+1)}_i=p^{(j+2)}_i$.

	It follows by induction that, for all $j\in J$ and all $i\in I_j$, we have
	$\lambda(p^{(j)}_i)\le \lambda(p^{(1)}_0)/2^i = \lambda(\origin)/2^i$. 
	Note that this bound is independent of the index $j$. Therefore
	\[
		\forall j\in J, ~ i\in I_j:~~~
			\len(\gamma^{(j)}_{i+1})\le 2(5\cdot 2^j)^d\cdot \lambda(p^{(j)}_i)
			\le 2(5\cdot 2^j)^d \cdot \frac{\lambda(\origin)}{2^i}.
	\]
	
	Let $\gamma$ be the concatenation of the paths $\gamma^{(j)}_{i+1}$,
	in the same order as we constructed them: first $\gamma^{(1)}_{i+1}$ for 
	increasing $i\in I_1$, then $\gamma^{(2)}_{i+1}$ for increasing $i\in I_2$,
	and so on until we reach the infinite sequence $\gamma^{(k)}_{i+1}$ 
	for increasing $i\in I_k$.
	We then have
	\begin{align*}
		\len(\gamma) ~&=~ 
		\sum_{j=1}^{k} \sum_{i\in I_j}\len(\gamma^{(j)}_{i+1}) 
		~\le~ \sum_{j=1}^{k} \sum_{i\in I_j}\left( 2(5\cdot 2^j)^d\cdot \frac{\lambda(\origin)}{2^i}\right).
	\end{align*}
	Using that for all $j<k$ the sets $I_j$ and $I_{j+1}$ have only $i_j=\max (I_j)$
	in common, that $2^j\le c^*$ for all $j<k$, that $2^{k}< 2c^*$, and that $\lambda(\origin) \le c^* \cdot |\origin \target|$,
	we get
	\begin{align*}
		\len(\gamma) ~&\le~ \sum_{j=1}^{k-1} \left( 2(5\cdot 2^j)^d\cdot \frac{\lambda(\origin)}{2^{i_j}} \right)+ \sum_{i=0}^\infty \left( 2(5\cdot 2c^*)^d\cdot \frac{\lambda(\origin)}{2^i}\right)
		\\ &\le~ 
			2\cdot 5^d\cdot \lambda(\origin) \cdot \sum_{j=1}^{k-1} (2^d)^j + 4(10 c^*)^d\cdot \lambda(\origin)
		\\ &=~ 
			2\cdot 5^d\cdot \lambda(\origin) \cdot \frac{(2^d)^k- 2^d}{2^d-1} + 4(10 c^*)^d\cdot \lambda(\origin)
		\\ &\le~ 
			2\cdot 5^d\cdot \lambda(\origin) \cdot \frac{(2c^*)^d - 2^d}{2^d-1} + 4(10 c^*)^d\cdot \lambda(\origin)
		\\ &\le~ 
			2\cdot 5^d\cdot \lambda(\origin) \cdot 2(c^*)^d + 4(10 c^*)^d\cdot \lambda(\origin)
		\\ &\le~ 
			4\cdot 5^d\cdot (c^*)^d \cdot \lambda(\origin) \cdot (1+2^d)
		\\ &\le~ 
			4\cdot 5^d\cdot (c^*)^d \cdot \lambda(\origin) \cdot (\tfrac{3}{2}\cdot 2^d )
		\\ &\le~ 
			6\cdot 10^d \cdot (c^*)^{d+1} \cdot |\origin\target|.
		\end{align*}

	The path makes an infinite number of straight-line moves.
	Since for each $i\in I_k$, the suffix of the path $\gamma$ after $p^{(k)}_{i+1}$
	is at distance at most $2\cdot \lambda(p^{(k)}_{i+1}) \le 
	2 \cdot \lambda(\origin)/2^{i+1}= \lambda(\origin)/2^i$ 
	from $\target$, Lemma~\ref{le:continuous} implies that $\target$ 
	is the endpoint of the path $\gamma$.
\end{proof}

\subsection{Extensions}
\label{sec:Extensions}

First, the same approach works for spaces of bounded doubling dimension
as long as there is a concept of path to connect points such that
the length of the path is the same as the distance between the points. 
We could also change the setting so that the cost
of moving from one point to another is the distance between the points. 
For this, one has to use $\eps$-nets in spaces of bounded 
doubling dimension~\cite{Har-PeledM06}.

When the target is known to have integral coordinates, we can modify 
the search so that it has a finite number of segments.
Indeed, as soon as we reach a point $p_i$ such that $\lambda(p_i)< 1/2$,
we know that the target is at distance smaller than $1/2$ from $p_i$ and
there is a unique point with integral coordinates in $B(p_i,\lambda(p_i))$.
We can then just move to that point. A similar approach works when
the target is known to have coordinates with bounded resolution by scaling
the setting. In the case where we can detect the target when we are within a given distance $\delta$, we can finish the search when we reach a point $p_i$ with $\lambda(p_i) \leq \delta$, thus, also bounding the
number of steps (which will depend on the initial estimate $\lambda(p_0)$).

Our strategies work also in the case where predictions may additionally underestimate the actual distance to the target by, say, some factor $c' \leq 1$, i.e., $c' \cdot |p\target|\le \lambda(p)$. Then, when both factors $c, c'$ are known, this is equivalent to scaling up the prediction by $1/c'$ to get a new one with factor $c/c'$. When the factors are unknown, the strategy in Theorem~\ref{thm:unknownc} still visits points with geometrically decreasing predictions and since for each point, say $p_i$, we now have that $|p_i\target| \leq \lambda (p_0)/(2^i \cdot c')$ the path converges to the target (as $c'$ is constant). 

One can consider the version when the target is an unknown $k$-flat $F$ 
in $\RR^d$. The same strategies work, 
also if we do not know the prediction factor nor the dimension, $k$, of the flat.
Indeed, we are constructing a path as a concatenation of segments
such that the distance to $F$ gets arbitrarily small for each suffix of the path.
Formally, for each $\eps>0$, there is a suffix of the path such that all
the points on the path are at distance $\eps$ from $F$.
Since the path is continuous and has bounded length, the endpoint of the path
has to be at distance $0$ from $F$.

Finally, note that for small $d$ there are tighter bounds on the size
of $\eps$-nets translating to better constants in our approach.
For $c\approx 1$ one can also exploit that we need an $\eps$-net of the spherical 
shell $S(\origin,r/c,r)$, which is much smaller than the whole ball $S(\origin,r)$. However, this improvement is expected to be small since even the $\eps$-nets
for spheres are not much better.

\section{Lower bound}
\label{sec:lower_bound}

In this section we provide a lower bound for the search problem
with $c$-predictions. Our lower bounds are meaningful 
when we assume that $c$ is large enough.
The idea is to construct $c$-predictions for infinitely many different
targets with the property that any two of them are indistinguishable, 
unless we are quite close to the target.
For each such prediction $\lambda$, there is a small ball around
the target where the value of $\lambda$ is different from the other
predictions, but most of the other predictions have the same value
on that small ball. Then, the searcher has
to visit \emph{all} the small balls around the targets,
because in the worst case the target is going to be
in the last small ball that is visited.

First, inspired by the technique by Elbassioni, Fishkin and Sitters~\cite{ElbassioniFS09} 
developed for the Euclidean Traveling Salesperson with Neighbourhoods, 
one can show the following bound.
See Dumitrescu and T{\'{o}}th~\cite{DumitrescuT16}
for an improvement over~\cite{ElbassioniFS09} where the same idea is
reused. 
  
\begin{lemma}
\label{lem:TSP}
	For any given radius $\delta\le 1/2$ and dimension $d\ge 1$, 
	consider the infinite family
	of balls 
	\[ 
		\BB ~=~ \big\{ B(p,\delta)\subset\RR^d\mid p\in B(\origin,1/2-\delta)\big\} 
			~=~ \big\{ B(p,\delta)\subset\RR^d\mid B(p,\delta) \subseteq B(\origin,1/2)\big\}.\]
	Each path in $\RR^d$ that contains 
	at least one point from each ball of $\BB$ has length at least
	$\left(\big(\frac{1}{2\delta}-1 \big)^d - 1\right) \cdot \delta \cdot \sqrt{\frac{\pi}{d}}$.
\end{lemma}
\begin{proof}
	For each $r\ge 0$, let $V_d(r)$ denote the volume of the $d$-dimensional 
	ball of radius $r$.
	It is clear that $V_d(r)=r^d\cdot V_d(1)$.
	
	Consider any path $\pi$ that touches each ball of $\BB$ and let $L$ be its length.
	Consider the volume of the Minkowski sum $\pi\oplus B(\origin,\delta)$
	of the (points on the) path $\pi$ and the ball $B(\origin,\delta)$, that is
	the set of points at distance at most $\delta$ from $\pi$.
	Since the path touches each ball of $\BB$,
	the set $\pi\oplus B(\origin,\delta)$ contains the center of each ball of $\BB$.
	This means that $\pi\oplus B(\origin,\delta)$ contains $B(\origin,1/2-\delta)$,
	and therefore we have 
	\[
		\vol \big( \pi\oplus B(\origin,\delta) \big) ~\ge~ 
			\vol\big( B(\origin,1/2-\delta)\big) ~=~ 
			\left(\frac{1}{2}-\delta\right)^d\cdot V_d(1).
	\]
	On the other hand, for each path $\pi$ of length $L$, we have 
	\[
		\vol \big( \pi\oplus B(\origin,\delta) \big) ~\le~ V_d(\delta) + L \cdot V_{d-1}(\delta)
		~=~ \delta^d \cdot V_d(1) +  L \cdot \delta^{d-1} \cdot V_{d-1}(1),
	\]
	with equality if and only if $\pi$ is a straight line. 
	See for example Lemmas~5.1 and~5.2 in~\cite{DumitrescuT16}.
	Therefore we get
	\[
		L ~\ge~ \frac{\left(\frac{1}{2}-\delta\right)^d\cdot V_d(1) - \delta^d \cdot V_d(1)}{\delta^{d-1} \cdot V_{d-1}(1)}
		  ~=~ \left(\Big(\frac{1}{2\delta}-1 \Big)^d - 1\right) \cdot \delta \cdot \frac{V_d(1)}{V_{d-1}(1)}.
	\]
	Using standard formulas for the volume of the $d$-dimensional ball, one gets
	that 
	\[
	\frac{V_d(1)}{V_{d-1}(1)} = \frac{\pi^{d/2}/\Gamma(d/2 + 1)}{\pi^{(d-1)/2}/\Gamma((d-1)/2 + 1)} = \pi^{1/2} \frac{\Gamma(d/2 + 1/2)}{\Gamma(d/2 + 1)} > \pi^{1/2} \frac{\sqrt{2}}{\sqrt{d+\frac{1}{2}}} \ge \sqrt{\frac{\pi}{d}},
	\]
	where we have used the inequality by Kershaw~\cite{Kershaw83} 
	\[
	\forall s\in (0,1),~ \forall x\ge 1:~~~
		\left( x+\frac{s}{2} \right)^{s-1} <~ \frac{\Gamma(x+s)}{\Gamma(x+1)}
	\]
	for $x=d/2$, which is at least $1$ for $d\ge 2$, and for $s=1/2$, which gives
	\[
		\left( \frac{d+1/2}{2} \right)^{-1/2} ~<~ \frac{\Gamma(d/2+1/2)}{\Gamma(d/2+1)}.
	\]	
	(For $d=1$ we can just check that $\Gamma(1) = 1 > \frac{\sqrt{\pi}}{2}=\Gamma(3/2)$.)
    The result follows.
\end{proof}

Next, fix a value $c\ge 2$. For each $\target\in B(\origin,1/2 - 1/c)$, let $\lambda_{\target}$ 
be the function defined as follows;
see Figures~\ref{fig:lower1} and~\ref{fig:lower2} for intuition.

\begin{equation}
	\lambda_{\target} (p) = \begin{cases}
						c \cdot |p\target|, &\text{ if $p\in B(\target, 1/c)$},\\
						1, &\text{ if $p\in B(\origin,1/2)\setminus B(\target, 1/c)$},\\
						2\cdot |p\origin|, &\text{ if $p\notin B(\origin,1/2)$}.
					\end{cases}
\label{eq:lambda}
\end{equation}

\begin{figure}[t]\centering
	\includegraphics[page=5,scale=.66]{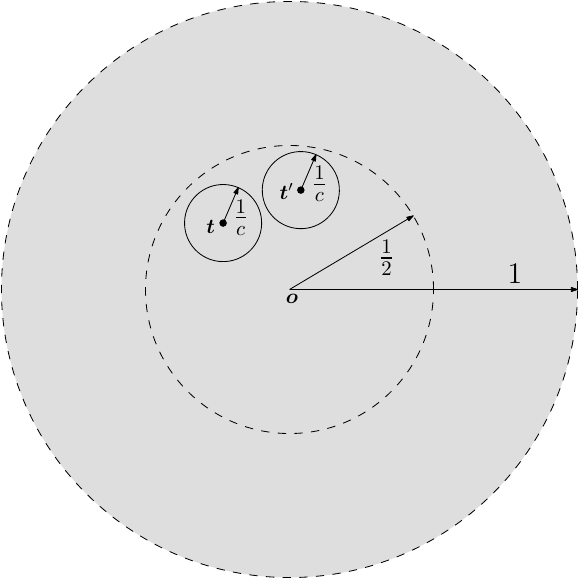}
	\caption{Parts in the domains for $\lambda_{\target}$ and $\lambda_{\target'}$
		for two targets when $d=2$.}
	\label{fig:lower1}
\end{figure}

\begin{figure}[t]\centering
	\includegraphics[page=7,width=\textwidth]{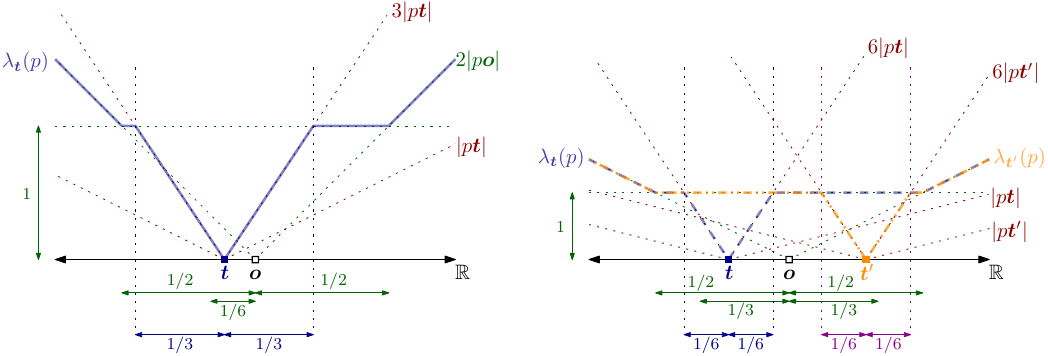}
	\caption{Examples of functions $\lambda_{\target}$ for $d=1$. Note that the axes
		have different scales.
		Left: example for $c=3$; the target has to be at distance at most $1/6$ from the origin. 
		Right: two functions for $c=6$; the target has to be at distance at most $1/3$ from 
		the origin.}
	\label{fig:lower2}
\end{figure}

\begin{lemma}
\label{lem:valid}
	Assume that $c\ge 2$ and $|\origin\target| \le \tfrac{1}{2}- \tfrac{1}{c}$.
	Then the function $\lambda_{\target}$ is a $c$-prediction for the target $\target$.
	Moreover, for any two distinct $\target,\target'$ satisfying the hypothesis,
	the functions $\lambda_{\target}$ and $\lambda_{\target'}$ agree on all
	points outside $B(\target, 1/c)\cup B(\target', 1/c)$.
\end{lemma}
\begin{proof}
	Consider any fixed $\target\in B(\origin,1/2- 1/c)$. 
	We have to show that $\lambda_{\target}$ satisfies 
	\[
		\forall p\in \RR^d: ~~~ |p\target| ~\le~ \lambda_{\target}(p) ~\le~ c \cdot |p\target|.
	\]
	We do this by considering points in each part of the domain used to define $\lambda_{\target}$.
	
	The first part is easy as $|p\target| < \lambda_{\target}(p) = c \cdot |p\target|$ for every $p\in B(\target, 1/c)$.
	For the second part we note that:
	\begin{align*}
		\forall p\in B(\origin,1/2)\setminus B(\target, 1/c) &: ~~~
				|p\target| ~\le~ |p\origin| + |\origin \target| ~\le~ 
					\frac{1}{2} + \frac{1}{2} ~=~ 1 ~=~ \lambda_{\target}(p).\\
		\forall p\in B(\origin,1/2)\setminus B(\target, 1/c) &: ~~~ 
				\lambda_{\target}(p) ~=~ 1 ~=~ c\cdot \frac{1}{c} ~\le~ c \cdot |p\target|.
	\end{align*}
	For the last part we first note:
	\begin{align*}
		\forall p\notin B(\origin,1/2) &: ~~~
				|p\target| ~\le~ |p\origin| + |\origin \target| ~\le~ 
					|p\origin| + \frac{1}{2} ~\le ~ 2\cdot |p\origin| ~=~ \lambda_{\target}(p).
	\end{align*}
	For each $p\notin B(\origin,1/2)$, let $q$ be the point 
	where the segment $p\target$ crosses the boundary of $B(\origin,1/2)$.
	Thus, $|q\origin|=1/2$ and, since $B(\target,1/c)\subseteq B(\origin,1/2)$,
	we also have $|q\target|\ge 1/c$. 
	We then have, using that $c\ge 2$,
	\begin{align*}
	\forall p\notin B(\origin,1/2) : ~~~ 
				\lambda_{\target}(p) ~&=~ 2\cdot |p\origin| ~\le~ 
					2\cdot \big(|pq|+|q\origin|\big) ~=~ 2\cdot |pq| + 1 \\
					~&\le~ 2\cdot |pq| + c\cdot |q\target| ~\le~ c\cdot \big(|pq|+|q\target|\big)
					~=~ c\cdot |p\target|.
	\end{align*}
	This covers all cases for the domain of $\lambda_{\target}$ and concludes the proof that
	$\lambda_{\target}$ is a $c$-prediction.
	
	From the definition of $\lambda_{\target}$ it is clear that
	$\lambda_{\target}(p)=\lambda_{\target'}(p)$ for all $p\notin B(\target,1/c)\cup B(\target',1/c)$
	because its value is independent of $\target$ and $\target'$.
\end{proof}

\begin{theorem}\label{thm:lower}
	Assume that $c> 2$.
	There is family of $c$-predictions for targets in $B(\origin,1/2-1/c)$ such that
	any search path in $\RR^d$ that uses $c$-predictions to find the target 
	has length at least $\left(\frac{c-2}{2}\right)^d  \cdot \frac{1}{c} \cdot \min\{\sqrt{\pi/d},1\}$. 
\end{theorem}
\begin{proof}
	For each $\target\in B(\origin,1/2 - 1/c)$, 
	consider the function $\lambda_{\target}$ defined
	in equation~\eqref{eq:lambda}. 
	Lemma~\ref{lem:valid} shows that the function
	$\lambda_{\target}$ is a $c$-prediction for each $\target\in N$.
	Consider the set of prediction functions 
	$\Lambda = \{ \lambda_{\target}\mid \target\in B(\origin,1/2 - 1/c)\}$.	
	Let $\BB =\{ B(\target,1/c)\mid \target\in B(\origin,1/2 - 1/c)\}$. 
	Note that to distinguish between two functions 
	$\lambda_{\target},\lambda_{\target'}\in \Lambda$ we have
	to evaluate them at some point of $B(\target,1/c)\cup B(\target',1/c)$.
	
	Consider any search strategy when the $c$-prediction function 
	is selected from $\Lambda$ by an \emph{adversary}. It
	may use that the $c$-prediction is from $\Lambda$. 
	In particular, we know beforehand that the target $\target$ is 
	a point of $B(\origin,1/2 - 1/c)$, selected by the adversary.	
	We claim that the search path has to visit all the balls of $\BB$.
	Indeed, while there are two distinct points 
	$\target,\target' \in B(\origin,1/2 - 1/c)$ such that
	the balls $B(\target,1/c)$ and $B(\target',1/c)$ are not visited by the search path,
	we have $\lambda_{\target}(p)=\lambda_{\target'}(p)$ for all points $p$ along the path,
	and therefore the information collected cannot discern 
	whether the target is $\target$ or $\target'$.
	Thus, in the worst case, the search path has
	to visit all the balls of $\BB$ but one, to identify which point of 
	$B(\origin,1/2 - 1/c)$ is the target,
	and then still may have to move to that last ball. That is,
	we may assume that in the worst case the adversary chooses $\lambda_{\target}\in \Lambda$ 
	such that $B(\target,1/c)$ is the \emph{last} ball of $\BB$ visited by the search
	strategy. 
	After the searcher deduces at which point of $B(\origin,1/2 - 1/c)$ the target is, they  
	still have to move to the target.
	
	We conclude that the search path has to visit all balls of $\BB$, 
	and for the last ball it still has to travel $1/c$ to the center.
	Lemma~\ref{lem:TSP} (with $\delta=1/c$)
	implies that the search path has length at least
	\begin{align*}
		\frac{1}{c} + \left(\Big(\frac{c}{2}-1 \Big)^d - 1\right) \cdot \frac{1}{c}\cdot \sqrt{\frac{\pi}{d}} 
		~&\ge~
		\frac{1}{c} + \left(\Big(\frac{c-2}{2}\Big)^d - 1\right) \cdot \frac{1}{c} \cdot \min\{\sqrt{\pi/d},1\} \\
		&\ge~
		\left(\frac{c-2}{2}\right)^d  \cdot \frac{1}{c} \cdot \min\{\sqrt{\pi/d},1\}.	\qedhere	
	\end{align*}
\end{proof}

For the \emph{competitive ratio} we have to compare the length of the search path
to the distance to the target, which in the construction is at most $1/2 - 1/c$.
We then obtain the following bound, which is interesting when $d\ge 2$ and $c > 2$.

\begin{corollary}
\label{col:bound}
	Consider the search with predictions problem in $\RR^d$ with prediction factor $c>2$.
	Any search strategy to reach the target has competitive ratio
	at least $\left(\frac{c-2}{2}\right)^{d-1}\cdot \min\{\sqrt{\pi/d},1\}$.
\end{corollary}
\begin{proof}
	Consider the construction of Theorem~\ref{thm:lower}.
	Since the search path has length at least 
	$\left(\frac{c-2}{2}\right)^d  \cdot \frac{1}{c} \cdot \min\{\sqrt{\pi/d},1\}$
	and the target is $\target\in B(\origin,1/2-1/c)$, the competitive ratio is at least
	\begin{align*}
		\frac{\left(\frac{c-2}{2}\right)^d  \cdot \frac{1}{c} \cdot \min\{\sqrt{\pi/d},1\}}{|\origin\target|}
	&\ge~
		\frac{\left(\frac{c-2}{2}\right)^d  \cdot \frac{1}{c} \cdot \min\{\sqrt{\pi/d},1\}}{1/2-1/c} \\
	&=~
		\frac{\left(\frac{c-2}{2}\right)^d  \cdot \frac{1}{c} \cdot \min\{\sqrt{\pi/d},1\}}{\frac{c-2}{2c}}\\
	&=~		
		\left(\frac{c-2}{2}\right)^{d-1} \cdot \min\{\sqrt{\pi/d},1\}.\qedhere
	\end{align*}
\end{proof}

To get a bound that is easier to grasp, let us assume that $c\ge 4$. 
In this case $(c-2)/2\ge c/4$ and we get the following bound.

\begin{corollary}
\label{col:bound2}
	Consider the search with predictions problem in $\RR^d$ with prediction factor $c\ge 4$.
	Any search strategy to reach the target has competitive ratio
	at least $(c/4)^{d-1}\cdot \min\{\sqrt{\pi/d},1\}$.
\end{corollary}

The lower bounds still hold in the case where the
target has integral coordinates or is detected
when we are within a given distance, with a slightly worse
constant in the basis of the exponential function. For this, we can just scale
the input such that the balls in the construction are large
enough to contain points with integer coordinates. Then,
we still have to visit all balls in order to detect which one has the target. 

A similar lower bound holds for randomized search strategies. 
In our construction, the searcher has to visit some balls
in the family $\BB$ of balls of Lemma~\ref{lem:TSP} (with $\delta=1/c$), and 
they can deduce something only when visiting the ball 
whose center is the target. In the worst-case (or adversary) 
model, the target is going to be the center of the last ball of $\BB$ 
that is visited. 
In a randomized setting, we select the target \emph{uniformly} at random 
among all the possible ones. This means that $\target$ is selected uniformly
at random among the points of $B(\origin,1/2-1/c)$.
When the search path has touched half of the balls of $\BB$
(in a measure theoretic sense), with probability
$1/2$ the ball $B(\target,1/c)$ has been touched.
A similar argument as the one used in the proof of Lemma~\ref{lem:TSP}
implies that any search path $\pi$ that touches half of the balls of $\BB$ 
must satisfy the volume inequality
\[
	\vol \big( \pi\oplus B(\origin,\delta) \big) ~\ge~ 
		\frac{1}{2}\cdot \vol\big( B(\origin,1/2-\delta)\big).
\]
With this we obtain a similar lower bound for randomized search
strategies, where the competitive ratio is about $1/4$ worse.

\section{Conclusions and research directions}

We have introduced and studied the problem of searching for a target in $\RR^d$ under a predictions model in which the searcher is given, at each position they visit, an approximate distance to the target. We presented strategies with competitive ratio $4\cdot 5^d \cdot c^{d+1}$ when the prediction factor $c$ is known, 
and $6\cdot 10^d\cdot c^{d+1}$ when the prediction factor $c$ is unknown.
We also showed a lower bound of $(c/4)^{d-1} \cdot \min\{\sqrt{\pi/d},1\}$ for the competitive ratio of any strategy under this model, assuming that $c\ge 4$, even when the prediction factor $c$ is known.

It will be interesting to revisit linear search and search in concurrent rays under this paradigm. Here, one can of course use the well-known linear search strategies without predictions, but we expect that predictions throughout the whole search will lead to a better competitive ratio.
Previous works~\cite{BoseCD_TCS15,HipkeIKL_DAM99,Lopez-OrtizS01} where a prediction is only given at the start, already provided a starting ground for this and showed that better strategies are possible. 

The case of $c=1+\eps$ for sufficiently small $\eps>0$ seems also interesting, because in this
case the spherical shells are very thin, and perhaps one can provide better bounds

For higher dimensions, the most natural question is how to bring the upper and lower bounds closer. 
Considering the problem in other metric spaces, such as graphs, and searching for more complex targets, such as convex subsets or $k$-flats, seems also a fruitful line of research. For example, does searching for a $k$-flat in $\RR^d$ with predictions behaves like searching for a point in $\RR^d$ or like a point in $\RR^{d-k}$?

One can also study whether weaker assumptions on the prediction function suffice to 
achieve constant competitive ratio.
For example, if the prediction $\lambda(\cdot)$ satisfies 
$|p\target|\le \lambda(p)\le 2\cdot |p\target|+|p\target|^2$ for all $p\in \RR^d$,
can one still achieve constant competitive ratio?
In the spirit of~\cite{BanerjeeC-AGL_ITCS23}, can we afford having completely wrong prediction 
in some region of $\RR^d$ with non-zero Lebesgue measure, if the region does not contain the target?

\subsection*{Acknowledgements}
The authors would like to thank Kostas Tsakalidis for suggesting the area of algorithmic problems with predictions and Konrad Swanepoel for providing a simple argument for the proof of Lemma~\ref{lem:bound}.

This research was funded in part by the Slovenian Research and Innovation Agency
(P1-0297, J1-2452, N1-0218, N1-0285), in part by the Erd{\H o}s Center, and in part by the European 
Union (ERC grant KARST, 101071836; ERC grant 882971, GeoScape).
Views and opinions expressed are
however those of the authors only and do not necessarily reflect those
of the European Union or the European Research Council.  Neither the
European Union nor the granting authority can be held responsible for
them.

\bibliography{main}
\end{document}